\documentclass[conference]{ifacconf}
\usepackage{natbib}
\usepackage{theorem}
\usepackage{amsmath}
\usepackage{amssymb}
\usepackage{amsfonts}

\usepackage[pdftex]{graphicx}
\usepackage{color}
\usepackage{multirow} 
\usepackage{algorithm} 
\usepackage{algpseudocode}

\let\theoremstyle\relax

\usepackage{float}
\usepackage{mathtools}
\usepackage{atbegshi}
\usepackage[font=small,justification=justified]{caption}
\usepackage{breqn}
\usepackage{amsthm}
\usepackage[utf8]{inputenc}
\usepackage[english]{babel}
\usepackage{varwidth}
\usepackage{comment}
\DeclarePairedDelimiter{\ceil}{\lceil}{\rceil}

\mathtoolsset{showonlyrefs}

\theoremstyle{definition}

\newtheorem{problem}{Problem}
\newtheorem{proposition}{Proposition}[section]
\newtheorem{definition}{Definition}[section]
\newtheorem{lemma}{Lemma}[section]

\newcounter{rem} 
\usepackage{resizegather}

\begin{document}

\begin{frontmatter}
\title{Persistent Surveillance With Energy-Constrained UAVs and Mobile Charging Stations}

%\thanks[footnoteinfo]{}
\author[First]{Sepehr Seyedi}, 
\author[Second]{Yasin Yaz{\i}c{\i}o\u{g}lu},  and 
\author[First]{Derya Aksaray}

\address[First]{Department of Aerospace Engineering and Mechanics, University of Minnesota, Minneapolis, Minnesota 55455.  (Email: seyed011@umn.edu,  daksaray@umn.edu)}

\address[Second]{Department of Electrical and Computer Engineering, University of Minnesota, Minneapolis, Minnesota 55455. (Email: ayasin@umn.edu)}

\begin{abstract}                
% Abstract of not more than 250 words.
\;  We address the problem of achieving persistent surveillance over an environment by using energy-constrained unmanned aerial vehicles (UAVs), which are supported by unmanned ground vehicles (UGVs) serving as mobile charging stations. Specifically, we plan the trajectories of all vehicles and the charging schedule of UAVs to minimize the long-term maximum age, where age is defined as the time between two consecutive visits to regions of interest in a partitioned environment. We introduce a scalable planning strategy based on 1) creating UAV-UGV teams, 2) decomposing the environment into optimal partitions that can be covered by any of the teams in a single fuel cycle, 3) uniformly distributing the teams over a cyclic path traversing those partitions, and 4) having the UAVs in each team cover their current partition and be transported to the next partition while being recharged by the UGV. We show some results related to the safety and performance of the proposed strategy. %results that the proposed strategy 1) always guarantees the safety of UAVs (they never run out of fuel during the mission) and 2) results in a long-term maximum age that is a function of the length of the cyclic path and the number of teams. The performance of the proposed strategy is also demonstrated via simulations.
%; 2) performance always land on charging stations run out of fuel under the proposed strategy always guarantees the UAVs to safely land on a charging station at any time during the mission and ) a performance result, that is the long-term maximum age, under the proposed strategy. The performance of the proposed strategy is also demonstrated via simulations.
\end{abstract}


\end{frontmatter}

\section{Introduction}
\vspace{-4mm}

In persistent surveillance missions, the goal is to continuously and repetitively obtain information about the activities or resources in a particular area. When surveillance is performed by a team of vehicles (agents), coverage control methods can be used to have the agents spread out and monitor the entire area (e.g., \cite{Cortes04,Pimenta08,Yazicioglu17TCNS}). However, if the environment is large compared to the sensing range of agents or the number of vehicles are insufficient, a full coverage may be infeasible. In such cases, a prominent approach is to partition the surveillance area into regions and try to minimize the maximum age, i.e.,  the time between consecutive visits to any of the regions. In the literature, various strategies were proposed to find agent trajectories that optimize such performance measures. Auction algorithms were used in \cite{nigam2008} to achieve region assignment among the agents to minimize the maximum age in the environment. A minimal cyclic path was found in \citep{elmaliach2009}, and the robots were located on the path to obtain uniform frequency of visiting the viewpoints. Trajectories that minimize the sum of ages while satisfying the individual temporal logic constraints of agents were computed in \citep{aksaray2015}.    

Some studies also considered persistent surveillance with energy-constrained agents, which return to a stationary base for recharging. An offline path-planning algorithm for maintaining a safe distance between energy-constrained UAVs and stationary base stations was developed in \citep{scherer2016}. A distributed energy-aware control policy for networked systems was proposed in \citep{derenick2011}. Temporal logic constraints were employed in \citep{aksaray2016, aksaray2015} to ensure that energy-constrained UAVs visit charging stations prior to losing power.

There is limited literature on persistent surveillance with energy-constrained unmanned aerial vehicles (UAVs) and mobile charging stations (e.g., unmanned ground vehicles (UGVs) equipped with charging stations). Mobile stations mainly enable the UAVs to be recharged as they are transported to another location. Without such a capability, persistent monitoring of a large environment requires placing multiple stationary charging stations such that every region in the environment is reachable from a charging station given the energy limitations of UAVs. On the other hand, environments of any size can be persistently monitored with mobile charging capability, even with one UAV and one UGV. The heuristics solvers for the Traveling Salesman Problem (TSP) were used to optimize the trajectories of multiple mobile charging stations servicing a team of UAVs with predetermined trajectories in \citep{smith2015}. A frugal feeding heuristic was used to optimize the trajectory of a single mobile charging station servicing multiple mobile robots with respect to locomotion costs in \citep{litus2009}. The problem of visiting a predetermined set of sites in the least amount of time was framed as a generalized TSP for multiple stationary and/or mobile charging stations in \citep{yu2018}. Different from the earlier works, we address the problem of simultaneously planning the trajectories of multiple UAVs and multiple UGVs for persistence surveillance missions. We present a scalable strategy based on optimally partitioning the environment and having uniform teams of a single UGV and multiple UAVs that patrol over a cyclic route of the partitions.

\section{Problem Statement}

\subsection{Preliminaries and Notation}
The set of real numbers is represented by ${\mathbb{R}}$, with the subset ${\mathbb{R}_+\subset \mathbb{R}}$ representing the set of non-negative real numbers. The set of natural numbers is ${\mathbb{N}}$. The vector of zeros with the is denoted by $\bold{0}$. For a real number ${r\in \mathbb{R}}$,  ${\lceil r \rceil}$ represents the ceiling of $r$. Similarly, $\lfloor r \rfloor$  represents the floor of $r$.
An undirected graph $\mathcal{G}(V,E)$ consists of a set of nodes $V$ connected by a set of edges $E \subseteq V \times V$. A Hamiltonian cycle over a graph $\mathcal{G}(V,E)$ is a cyclic path that visits each node $v\in V$ exactly once. 

\subsection{Problem Formulation}
\textit{Environment -}
We consider a persistent surveillance scenario in a 3D convex obstacle-free volume
\begin{equation}
\resizebox{.93\hsize}{!}{$\mathcal{Q}\coloneqq \{(x,y,z) \in \mathbb{R}^3 \vert 0 \le x \le x_{max}, 0 \le y \le y_{max}, z \ge 0\}$,}
\end{equation} 
where $x_{max},y_{max}\in \mathbb{R}_+$. We assume that ${n \in \mathbb{N}}$ homogeneous UAVs with limited energy move over $\mathcal{Q}$ and are supposed to patrol the area defined as
\begin{equation}
\resizebox{.93\hsize}{!}{$\mathcal{Q}_{\bar{z}} \coloneqq \{(x,y,z) \in \mathbb{R}^3 \vert 0 \le x \le x_{max}, 0 \le y \le y_{max}, z = \bar{z}\}$},
\end{equation}
in coordination with ${m\in \mathbb{N}}$ homogenous UGVs (i.e., mobile charging stations). Here, $\bar{z}\in \mathbb{R}_+$ can be determined based on the desired image resolution if it is an aerial imaging scenario.
% The patrolling task is formulated as a coverage problem for the UAVs over an area 
% \begin{equation}
% \mathcal{Q}_{\bar{z}} \coloneqq \{(x,y,z) \in \mathbb{R}^3 \vert x \le \bar{x}, 0 \le y \le \bar{y}, z = \bar{z}\}
% \end{equation}

Let each UAV have a square detection footprint with a dimension of ${d \in \mathbb{R}_+}$ when flying at an altitude of ${\bar{z}}$. We assume that the area defined by $\mathcal{Q}_{\bar{z}}$ can be discretized into a Cartesian grid where each grid cell is a square of dimension $d$. Accordingly, the grid will have ${\bar{x} \times \bar{y}}$ cells, where ${\bar{x} = \frac{x_{max}}{d}}$, and ${\bar{y} = \frac{y_{max}}{d}}$. The geometric centers of the cells collectively comprise a set of points ${V \coloneqq \{v_1,...,v_{|V|}\}}$, where each ${v_i \in \mathbb{R}^3}$ needs to be persistently visited by UAVs. These points in $V$ will be referred to nodes in the rest of the document.

\textit{Dynamics -}
At any time $t$, let the position of UAV $i \in \{1,\dots,n\}$ be denoted by ${p_i^A(t) = [x_i^A(t),y_i^A(t),z_i^A(t)] \in \mathcal{Q}}$ and the position of UGV $i \in \{1,\dots,m\}$ be denoted by \linebreak  ${p_i^G(t) =[x_i^G(t),y_i^G(t),0] \in \mathcal{Q}}$. Moreover, the positions \linebreak of all UAVs and all UGVs at time $t$ will be \linebreak  expressed by ${p^A(t) = [p_1^A(t),...,p_n^A(t)]^{\intercal}} \in \mathbb{R}^{n\times3}$ and ${p^G(t) = [p^G_1(t),...,p^G_m(t)]^{\intercal}} \in \mathbb{R}^{m\times3}$, respectively.%and let $p^G(t)$ be the array of all UGV positions   at time $t$. 
We will use the convention that ${\mathbf{p}^A_i}$ will represent the full trajectory of UAV $i$, $\mathbf{p}^G_i$ will represent the full trajectory of UGV $i$, and the corresponding sets of UAV and UGV trajectories will be denoted by $\mathbf{p}^A$ and $\mathbf{p}^G$, respectively.

The system of vehicles is subject to single integrator dynamics, with UAVs and UGVs having maximum speeds of ${u_{max}^A \in \mathbb{R}}$ and ${u_{max}^G \in \mathbb{R}}$, respectively. Accordingly, the dynamics of the vehicles can be written as
\begin{equation}\label{UAVdynamics}
    \dot{p}_i^A(t) = u_i^A(t), \qquad i=1,...,n ,
\end{equation}
\begin{equation}\label{UGVdynamics}
    \dot{p}_i^G(t) = u_i^G(t), \qquad i=1,...,m ,
\end{equation}
where ${u_i^A(t) \in \mathbb{R}^3}$ and ${u_i^G(t) \in \mathbb{R}^3}$ are the control inputs for UAV $i$ and UGV $i$ at time $t$, respectively. Moreover, these control inputs have the following constraints: 
\begin{equation}\label{Acon}
    \|u_i^A(t)\|_2 \le u_{max}^A  \; \text{ and } \; \|u_i^G(t)\|_2 \le u_{max}^G  \quad \forall i,
\end{equation}
depicting a maximum speed for both vehicle types;
\begin{equation}\label{Gcon2}
    proj_z\{u_i^G(t)\} = 0,
\end{equation}
indicating that the UGV is moving only on the $x-y$ plane; 
\begin{equation}\label{Acon2}
    u_i^A(t) = [0 \ 0 \ -u_{max}^A]  \; \forall i \; \text{ s.t. } \; z_i^A(t) \ge u_{max}^A\frac{e_i(t)}{\beta^-},
\end{equation}
enforcing that the UAVs do not to run out of energy during flight by imposing emergency landing when they reach a critical energy;
\begin{equation}
\resizebox{.88\hsize}{!}{$\begin{split}\label{Acon3}
    u^A_i(t) \in \{\mathbf{0}\} \cup\left\{ u^G_j(t) \ \middle| \  p^G_j(t)=p^A_i(t)\right\},
    \forall i \text{ s.t } e_i(t)=0, 
    \end{split}$}
\end{equation}
enforcing that the UAVs are either stationary or transported by a UGV when their energy is depleted;
\begin{equation}\label{Acon4}
\resizebox{.9\hsize}{!}{$\begin{split}
    u^A_i(t) \in & \{\mathbf{0}\}\cup \left\{ u \in \mathbb{R}^3\ \middle| \ proj_z(u)>0\right\} \cup \\ 
    & \left\{ u^G_j(t) \ \middle| \  p^G_j(t)=p^A_i(t)\right\} \; \forall i \; \text{ s.t. } z^A_i(t)=0,
    \end{split}$}
\end{equation}

enforcing that the UAVs on the ground can be stationary, take off (i.e., $proj_z(u)>0$), or be transported by a UGV. 
The set of control inputs to the UAVs and UGVs at time $t$ are denoted by ${u^A(t)=[u_1^A(t),...,u_n^A(t)]^{\intercal}} \in \mathbb{R}^{n\times 3}$ and ${u^G(t)= [u_1^G(t),...,u_m^G(t)]^{\intercal}}\in \mathbb{R}^{n\times 3}$, respectively. The full history of inputs to UAV $i$ and UGV $i$ will be denoted as $\mathbf{u}_i^A$ and $\mathbf{u}_i^G$, respectively. Similarly, the sets of all $\mathbf{u}_i^A$ and $\mathbf{u}_i^G$ will be denoted as $\mathbf{u}^A$ and $\mathbf{u}^G$, respectively.

Each UAV $i$ is assumed to have limited energy, ${e_i(t) \in [0,\bar{e}]}$, where ${\bar{e}\in \mathbb{R}}$ is the maximum energy capacity of all UAVs. Energy limitations of UGVs are neglected, and it is assumed that any UGV can simultaneously charge any number of UAVs. 
% Each UAV has a binary state $\alpha_i \in \{0,1 \}$, which denotes whether it is active ($\alpha_i=1$: flying) or inactive ($\alpha_i=0$: charging)
% \begin{equation}
%     \label{act}
%     \alpha_i(t)= \begin{cases}
%  0 & \quad \text{if charging}   \\
% 1  & \quad \text{otherwise}.
% \end{cases}
% \end{equation}
% in  which is equivalent to the maximum number of consecutive hops that a fully-charged UAV can take on $\mathcal{G}$,
The energy of UAV $i$, i.e., $e_i(t)$,  decreases at a constant rate ${\beta^- \in \mathbb{R}_+}$ when active, and increases at a constant rate ${\beta^+ \in \mathbb{R}_+}$ when charging,
$$
e_i(0)=\bar{e},
$$
\begin{equation} \label{energy}
\dot{e}_i(t) =
\begin{cases}
  \beta^+  & \  \text{if } \exists j: p_i^A(t)=p_j^G(t), \text{and } e_i(t) < \bar{e},\\
  -\beta^-  & \  \text{if } z^A_i(t) > 0,\\
  0 & \ \text{otherwise.}
\end{cases}
\end{equation}
% During the inactive phase, a UAV can only move if it is being charged and transported by a UGV, i.e., 
% \begin{multline}
% \label{inacmove}
%     \alpha_i(t)=0, \; p_i(t+\Delta t)\neq p_i(t) \Rightarrow \exists j: p_i(t+\Delta t)= q_j(t+\Delta t), \\  p_i(t)= q_j(t).
% \end{multline}

% For simplicity, it will be assumed that the UAVs and UGVs can travel to any adjacent node in $\mathcal{G}$ in one and $u \in \mathbb{N}$ time steps, respectively, i.e.,

% $$
% d(p_i(t+1), p_i(t)) \leq 1,  
% $$
% $$
% d(q_i(t+1), q_i(t)) \leq 1,
% $$
% \begin{equation}
% \label{mots}
% q_i(t+1)\neq q_i(t) \Rightarrow q_i(t)=  q_i(t-1) = \hdots =q_i(t-u-1).
% \end{equation}

\textit{Objective -} We define the objective of the problem based on a metric called age. The age of a node is defined as the difference between the current time and the last time there was a UAV present at that node. If the node is occupied by a UAV, the node's age is zero. Further, we assume that the ages of all nodes are initialized to zero at time $t=0$.
For any set of UAVs, the initial conditions $p^A(0)$, the control inputs $u^A(t)$ and $u^G(t)$ result in a set of UAV trajectories that lead to a set of times $\mathcal{A}(v,\mathbf{p}^A)$, which comprises the times when a UAV arrives at a node ${v\in V}$ that was previously unoccupied. Similarly, the UAV trajectories lead to a set of times $\mathcal{D}(v,\mathbf{p}^A)$, which comprises the times when a UAV departs from a node ${v\in V}$ that becomes unoccupied. For some non-zero ${\varepsilon \in \mathbb{R}_+}$, the arrival and departure time sets can be written as
\small
\begin{equation} \label{arrive}
    \mathcal{A}(v,\mathbf{p}^A) \coloneqq  \Bigg\{t\in (0,\infty) \ \biggm| \begin{array}{l} \ \exists i: p_i^A(t) = v \\ \forall i \  p_i^A(\tau) \neq v , \ \tau \in (t-\varepsilon,t)  \end{array} \Bigg\},
\end{equation}
\begin{equation}\label{depart}
    \mathcal{D}(v,\mathbf{p}^A) \coloneqq  \Bigg\{t\in (0,\infty) \ \biggm| \begin{array}{l} \ \exists i: p_i^A(t) = v \\ 
    \forall i \ p_i^A(\tau) \neq v, \tau\in(t,t+\varepsilon)  \end{array} \Bigg\},
\end{equation}
\normalsize
The maximum age over all nodes ${v \in V}$ will depend on the arrival and departure times of the UAVs as defined in \eqref{arrive} and \eqref{depart}. Let ${t_d \in \mathcal{D}(v,\mathbf{p}^A)}$ denote a departure time, and ${t_{\alpha}(v,t_d) \in \mathcal{A}(v,\mathbf{p}^A)}$ denote the arrival time to node $v$ that minimizes the difference ${(t_{\alpha}-t_d)}$. Then the maximum age over the time interval $(t,\infty)$ can be defined as follows,
\begin{equation}\label{age}
    \bar{T}(t) = \underset{v\in V, t_d \ge t}{\text{max}}\{t_{\alpha}(v,t_d)-t_d: v\in V, t_d \in \mathcal{D}(v,\mathbf{p}^A) \},
\end{equation}
\begin{equation}
    t_{\alpha}(v,t_d) = \text{min}\{ \tau \in \mathcal{A}(v,\mathbf{p}^A): \tau \ge t_d\}. \nonumber
\end{equation}

% Denote the age of a point $v \in V$ at time $t$ by $T_v(t)$,
% \begin{equation}\label{age}
%     T_v(t) = t - \alpha_t
% \end{equation}
% where $\alpha_t = \text{max}\{\alpha\in\mathcal{A}(v) \ \vert \ \alpha \le t\}$.

% For each node $v \in V$, let $T_v(t)$ represent the age of the node $v$ at time $t$, which is defined as the amount of time since the last visit to cell $v$ by an active UAV. Accordingly,
% \begin{equation}
%   \label{age}
%   T_v(t+\Delta t)= \begin{cases}
%  0 &  \text{if } \exists i: \alpha_i(t+\Delta t) p_i(t+\Delta t)=v,  \\
% T_v(t)+\Delta t&  \text{otherwise,}

% \end{cases}
% \end{equation}

\begin{problem} (Age Minimization)
For a given environment $\mathcal{Q}$ and a team of $n$ UAVs and $m$ UGVs with the energy and mobility constraints of \eqref{Acon} through \eqref{energy}, the time evolution of the node ages depend on the choices of ${\mathbf{u}^A,\mathbf{u}^G, p^A(0)}$ and $p^G(0)$. 
% which is the maximum of $T_v(t)$ over all $v\in V$ in the long run, i.e.,
% \begin{equation}
% \bar{T}(p,q,\alpha) := \underset{t\to \infty}{\text{lim sup}}  \ \left( \underset{v\in V}{\text{max}} (T_v(t))\right).
% \end{equation}
Find the trajectories of all vehicles (UAVs and UGVs) and the charging schedule of UAVs to minimize the maximum age, subject to the energy and mobility constraints of the vehicles, i.e.,
\begin{align}\label{eq:problem}
    \min_{p^A(0),p^G(0),\mathbf{u}^A,\mathbf{u}^G} & \quad \underset{t\to \infty}{\text{lim sup}} \left(\bar{T}(t)\right),\\
    s.t. & \quad \eqref{UAVdynamics}, \eqref{UGVdynamics},  \eqref{Acon},\eqref{Gcon2},
    \eqref{Acon2},\eqref{Acon3},\eqref{Acon4}, \eqref{energy}.\nonumber
\end{align}
\end{problem}

Such patrolling problems are typically NP-hard similar to the TSP, which requires to find shortest closed tours visiting all the nodes (e.g., \cite{pasqualetti2010,yu2018,laporte1992}). Exact solutions usually require searching among all the possibilities, which is not scalable. This has motivated the development of approximation algorithms that produce acceptably good solutions.

\section{Main Results}

We propose a scalable planning strategy for simulatenously finding the trajectories of UAVs and UGVs as follows: 1) creating teams that involve a single UGV and equal number of UAVs, 2) decomposing the environment into the largest rectangular partitions that can be covered by any of the teams in a single fuel cycle, 3) uniformly distributing the teams over a cyclic path traversing those partitions, and 4) having the UAVs in each team cover their current partition and be transported to the next partition on the cycle while being recharged by the UGV in their team.
\subsubsection{Partitioning}
The area $\mathcal{Q}_{\bar{z}}$ is partitioned using identical rectangular partitions containing ${a_1 \le \bar{x}}$ cells along the x-direction and ${a_2 \le \bar{y}}$ cells along the y-direction. Given an environment graph $\mathcal{G}=(V,E)$, we consider a partition set $\mathcal{P}$ which is comprised of rectangular partitions  $\mathcal{P}_i$ ($i=1,...,|\mathcal{P}|$) where each $\mathcal{P}_i$ induces a subgraph $ \mathcal{G}[\mathcal{P}_i] = (V_i,E_i)$ such that $\cup_i V_i = V$.
\begin{algorithm}\small
	\caption{Partitioning}\label{alg_4}
	\textbf{INPUT: } $\bar{x},\bar{y},a_1,a_2 $ \;
	\textbf{OUTPUT: } $\mathcal{P}$
	\begin{algorithmic}[1]
	\State $\mathcal{P} \gets \emptyset$, $i \gets 0$
	\For {$k_1=1:\ceil{\bar{x}/a_1}-1$}
	\For {$k_2=1:\ceil{\bar{y}/a_2}-1$}
	\State $i \gets i+1$
	\State  $\mathcal{P}_i \coloneqq  \left\{v\in V \ \middle| \ \begin{array}{l} (k_1-1)a_1 < x < k_1a_1 \\ 
        (k_2-1)a_2 < y < k_2a_2  \end{array}\right\}$
    
% 	\State $c_i \gets$ centroid of $\mathcal{P}_i$
	\EndFor
	\EndFor
	\For {$k=1:\ceil{\bar{y}/a_2}-1$}
	\State $i \gets i+1$
	\State $\mathcal{P}_i \coloneqq \left\{v(x,y)\in V \ \middle| \ \begin{array}{l} \bar{x}-a_1 < x < \bar{x} \\ (k-1)a_2 < y < ka_2 \end{array} \right\}$
% 	\State $c_i \gets$ centroid of $\mathcal{P}_i$
	\EndFor
	
	\For {$k=1:\ceil{\bar{x}/a_1}-1$}
	\State $i \gets i+1$
	\State $\mathcal{P}_i \coloneqq \left\{v(x,y)\in V \ \middle| \ \begin{array}{l} (k-1)a_1 < x < ka_1 \\ \bar{y}-a_2 < y < \bar{y} \end{array} \right\}$
% 	\State $c_i \gets$ centroid of $\mathcal{P}_i$
	\EndFor
	\State $i \gets i+1$
	\State $\mathcal{P}_i \coloneqq \left\{v(x,y)\in V \ \middle| \ \begin{array}{l} \bar{x}-a_1 < x < \bar{x} \\ \bar{y}-a_2 < y < \bar{y} \end{array} \right\}$
% 	\State $c_i \gets$ centroid of $\mathcal{P}_i$
\end{algorithmic}
\end{algorithm}
In Alg. \ref{alg_4}, there are four distinct sections. An example partition set is depicted in Fig. \ref{fig:partitioning}, where Fig. \ref{fig:partitioning}(a) corresponds to the lines $3-7$ of Alg. \ref{alg_4}, Fig. \ref{fig:partitioning}(b) to the lines $8-11$, Fig. \ref{fig:partitioning}(c) to the lines $12-15$, and Fig. \ref{fig:partitioning}(d) to the lines $16-17$.

\begin{figure}[h!]
    \centering
	\includegraphics[scale=0.42]{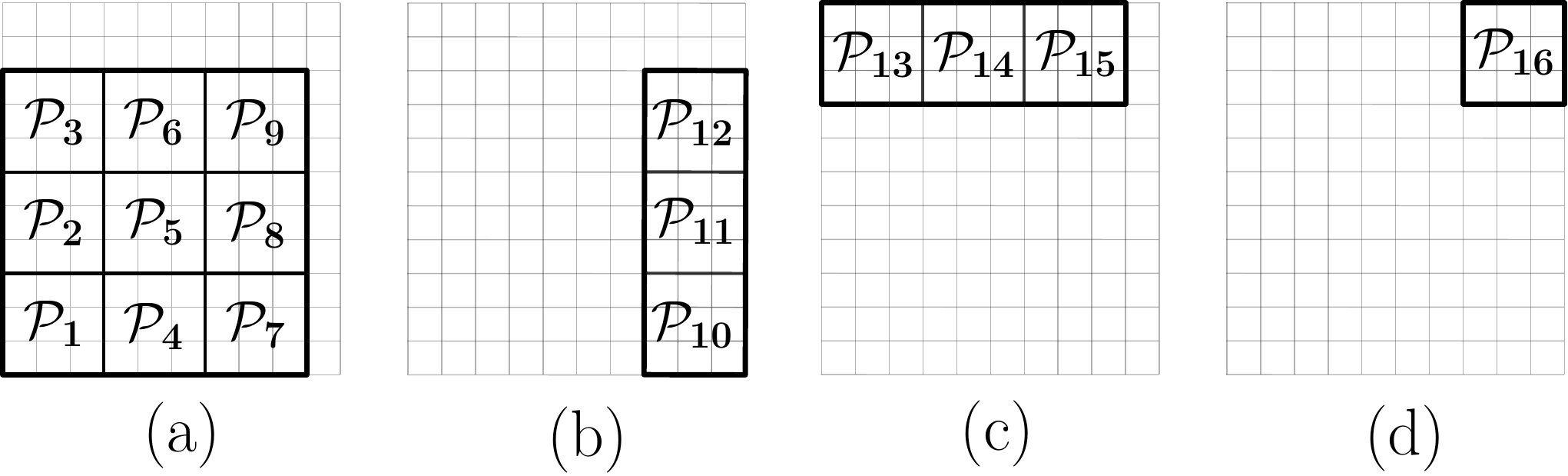}
	\caption{\small Illustration of the partitions generated via Alg. \ref{alg_4} for the case of $a_1=a_2=3$, $\bar{x}=10$, and $\bar{y}=11$.}
	\label{fig:partitioning}
\end{figure}

\subsubsection{Supercycles}

For simplicity, we assume that the number of UAVs $n$ is divisible by the number of UGVs $m$. Hence, $m$ teams are formed, each with $1$ UGV and $\frac{n}{m}$ UAVs.

% For each fuel cycle, a team of UGV $j$ and $\bar{n}$ UAVs is assigned to a partition $\mathcal{P}_i$. In one fuel cycle, the $\bar{n}$ UAVs are released from the release point at the center of the partition, the UAVs patrol all nodes in $\mathcal{P}_i$, UAVs return to the release point, and then during charging UGV $j$ moves the UAVs assigned to it from its current partition $\mathcal{P}_i$ to another partition $\mathcal{P}_h$. Extending this procedure until all partitions are covered, we consider a supercycle of fuel cycles over the mission space as in Fig. \ref{fig:supcycle}. Specifically, we consider a Hamiltonian super-cycle of fuel cycles over equivalent rectangular partitions of the environment, allowing for partitions to overlap if needed.

\begin{definition}\label{def:supercycle}
(Supercycle): For a given UAV-UGV team with given energy capacity, a supercycle of a partition set $\mathcal{P}$ of the area $\mathcal{Q}_{\bar{z}}$ is a sequence of UAV-UGV patrolling-charging cycles that patrols all partitions in $\mathcal{P}$ once.
\end{definition}
Definition \ref{def:supercycle} states the following: Consider a sequence of partitions $(\mathcal{P}_0, \dots, \mathcal{P}_f, \mathcal{P}_0)$ that starts and ends with partition $\mathcal{P}_0$ and contains all the partitions in $\mathcal{P}$ only once. Assume that the team starts at $\mathcal{P}_0$. As the UGV occupies the center of $\mathcal{P}_0$, all UAVs are released and visit all nodes in $\mathcal{P}_0$, and then return back to the release point. Once all UAVs have returned to the release point, the UGV carries all UAVs to the release point of the next partition while simultaneously charging the UAVs. If needed, the UGV will wait at the next release point until all UAVs have reached full energy. After all UAVs have returned to the release point at the final partition $\mathcal{P}_f$, the UGV moves the team back to the release point of the first partition $\mathcal{P}_0$, and the same cycle repeats. Such a cycle is called a supercycle and its period is denoted by $T_c$.

%\begin{enumerate}
%    \item As the UGV occupies the center of partition $\mathcal{P}_i$, all UAVs are released and visit all nodes in the partition, returning back to the release point. 
%    \item Once all UAVs have returned to the release point, the UGV carries all UAVs to the release point of the next partition while simultaneously charging the UAVs back to their full energy capacity.
%    \item The UGV then waits at the next release point until all UAVs have reached full energy, if needed.
%    \item After all UAVs have returned to the release point at the final partition in the supercycle, the UGV moves the team back to the first partition release point, and the cycle repeats.
%    \item The period for a supercycle is denoted $T_c$.
%\end{enumerate} 
% The process of this patrolling-charging cycle being carried out repeatedly over a Hamiltonian cycle of the partitions is referred to as a supercycle.

\begin{figure}[h!]
    \centering
	\includegraphics[scale=0.6]{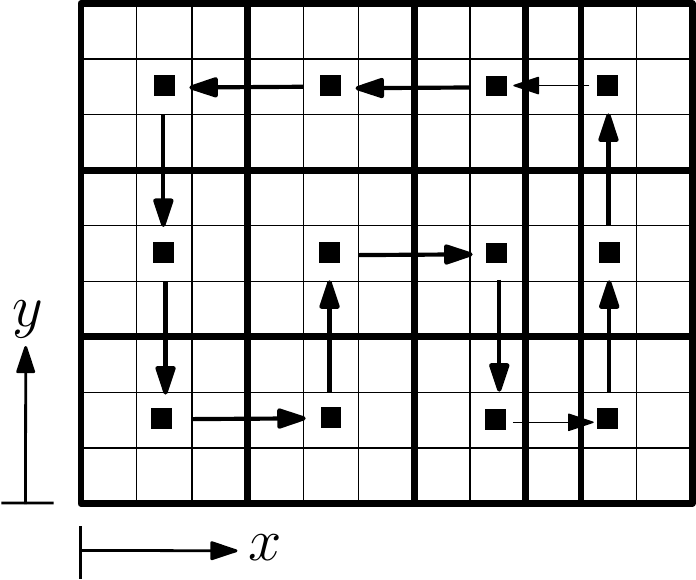}
	\caption{\small{Example of a supercycle for a square partition of size $3\times 3$ for an evironment of size $12X9$.}} %Note that the x- and y-coordinates are normalized to the grid cell dimensions in the x- and y-directions, respectively.}}
	\label{fig:supcycle}
\end{figure}
For a given supercycle, the UGV follows a Hamiltonian cycle that passes through the centers of each partition once. 
The UGV trajectories that define the supercycle are determined using the Dantzig-Fulkerson-Johnson (DFJ) algorithm \citep{dantzig1954solution}. The DFJ algorithm returns a sequence $\{\mathcal{P}_i\}_{i=1:|\mathcal{P}|}$ of partitions of the partition set $\mathcal{P}$ sorted according to the minimum-distance Hamiltonian cycle through the partitions.
Let $c_i \in \mathbb{R}^3$ be the center of partition $\mathcal{P}_i$, and let the index $i$ correspond to the order of the partition in the supercycle. Then the UGV must travel a distance of ${\|c_{i+1}-c_{i}\|}$ to carry its assigned UAVs from partition $\mathcal{P}_i$ to partition $\mathcal{P}_{i+1}$. In the first cycle, suppose the UAV-UGV team releases UAVs at partition $\mathcal{P}_i$ at time $t_i$ assuming the supercycle begins at time $t=0$. Then the velocity of the UGV when it is carrying the UAVs from $c_i$ to $c_{i+1}$ will be,
\begin{equation}\label{ugv_traj} 
    \small u^G(t) = 
        \begin{cases}
            {u_{max}^G}\frac{c_{i+1}-c_i}{\|c_{i+1}-c_i\|_2} & \text{if $t \in \bigcup_{i} \Big(t^{out}_{i,j},t^{in}_{i+1,j}\Big)$} \\
            0 & \text{otherwise}
        \end{cases}
\end{equation}
where $t^{out}_{i,j}$ is the time when the UGV starts to move from $\mathcal{P}_i$ to $\mathcal{P}_{i+1}$ in the j$^{th}$ supercycle, $t^{in}_{i+1}$ is the time when the UGV arrives the center of $\mathcal{P}_{i+1}$ in the j$^{th}$ supercycle, and they can be expressed as 
\begin{align}
t^{out}_{i,j} &= t_i + jT_c + \frac{\Delta e}{\beta^-} \\
t^{in}_{i+1,j} &= t_i+jT_c +\frac{\Delta e}{\beta^-}+ \frac{\|c_{i+1}-c_i\|_2}{u_{max}^G}
\end{align}
where $\frac{\Delta e}{\beta^-}$ is the required time for the UAVs to cover a partition from release to the return of the last UAV to the release point, $T_c$ is the period of the supercycle, and $j$ is the repetition index of the supercycle.   
%\small
%\begin{equation}
%    \forall t \in \bigcup_{j=0}^{\infty}\left(t_i+jT_c+\frac{\Delta e}{\beta^-} \ \scalebox{2}{,} \ t_i+jT_c +\frac{\Delta e}{\beta^-}+ \frac{\|c_{i+1}-c_i\|_2}{u_{max}^G}\right) ,
%\end{equation}
%\normalsize
%and will be zero otherwise. As this will happen during every supercycle in a cyclic manner, the time interval is the union of infinitely many intervals separated by $T_c$. As the teams follow a supercycle, $c_{|\mathcal{P}|+1} = c_1$.
% The amount of time required from the arrival of the last UAV to the release of the UAVs at the next partition will be the greater of the UGV transit time between partitions and the charging time required to replenish an energy difference of $\Delta e$. Therefore, 
Assuming the UGV moves at maximum speed while in motion, the amount of time required for the UAV-UGV team to go from releasing UAVs at partition $\mathcal{P}_i$ to releasing UAVs at the next partition $\mathcal{P}_{i+1}$ is,
\begin{equation}\label{eq:delta_time_partitions}
    \Delta t_i = \frac{\Delta e}{\beta^-}+ \text{max}\left(\frac{\|c_{i+1}-c_i\|_2}{u_{max}^G},\frac{\Delta e}{\beta^+}\right). 
\end{equation}
where $\beta^-$ and $\beta^+$ are the discharging and recharging rates. This result leads to a supercycle period of,
\small
\begin{equation}\label{objective}
\begin{split}
    T_c = |\mathcal{P}|\frac{\Delta e}{\beta^-}+\sum_{i=1}^{|\mathcal{P}|-1} \text{max}\left(\frac{\|c_{i+1}-c_i\|_2}{u_{max}^G},\frac{\Delta e}{\beta^+}\right)+
    \text{max}\left(\frac{\|c_1-c_{|\mathcal{P}|}\|_2}{u_{max}^G},\frac{\Delta e}{\beta^+}\right),
    \end{split}
\end{equation}
\normalsize
where the total number of partitions is given by $|\mathcal{P}| = \ceil[\Big]{\frac{\bar{x}}{a_1}}  \ceil[\Big]{\frac{\bar{y}}{a_2}}$.
%\begin{equation}\label{partsize}
%|\mathcal{P}| = \ceil[\Big]{\frac{\bar{x}}{a_1}}  \ceil[\Big]{\frac{\bar{y}}{a_2}}.
%\end{equation}

% For multiple homogeneous UAV-UGV teams, if the teams are following an identical supercycle, but are equally spaced along the supercycle in time, then the temporal spacing between the teams will be
% \begin{equation}\label{sup_period}
% \bar{T}_c = \frac{T_c}{m}.    
% \end{equation}

% \begin{figure}[h!]
% 	\centering
% 	\includegraphics[scale=0.58]{supercycle_examples.pdf}
% 	\caption{Examples of the three supercycle configurations being implemented: (a) Even number of columns, (b) Odd number of columns, even number of rows, and (c) Even number of rows and columns.}
% 	\label{fig:cycle_examples}
% \end{figure}
 
\subsubsection{UAV Trajectory Generation}
\begin{sloppypar}
Within each partition, the trajectories for the UAVs are planned such that all nodes in the partition are visited. In order to evenly distribute node assignments among the UAVs, a partition is further subpartitioned using a quadrant angle referenced to the release point and a reference axis, as depicted in Fig. \ref{fig:subpart}. Let $\{\nu_{i,j}\}_{j=1:|V_i|}$ be the sequence of nodes in partition $\mathcal{P}_i$ sorted according to their quadrant angle. The number of nodes allocated to the subpartition $k$ is $|V_{i,k}|=\ceil[\Big]{\frac{|V_i|-s_k}{\frac{n}{m}-(k-1)}}$, where $s_k = \sum_{j=1}^{k-1} |V_{i,j}|$, and $V_{i,k}\in V_i$ is the set of nodes belonging to subpartition $k$. The node set for subpartition $k$ is 
\begin{equation}\label{subpartition}
 V_{i,k}=\{\nu_{i,j} | s_{k}+1 \leq j \leq s_k+|V_{i,k}| \}\cup \{c_i\}   .
\end{equation}
% \begin{equation}
%     |V_{g,i}| = \ceil[\Bigg]{\frac{|V_g|-s_i}{\bar{n}-(j-1)}},
% \end{equation}
% \begin{equation}
%     s_i = \sum_{j=1}^{i-1} |V_{g,j}|
% \end{equation}
% and the node assignment is,
% \begin{equation}\label{subpartition}
%     V_{g,i} = \{\bar{V}_g(1+s_i),..., \bar{V}_g(|V_{g,i}|+s_i),c_i\}.
% \end{equation}
\end{sloppypar}

\begin{figure}[h!]
	\centering
	\includegraphics[scale=0.8]{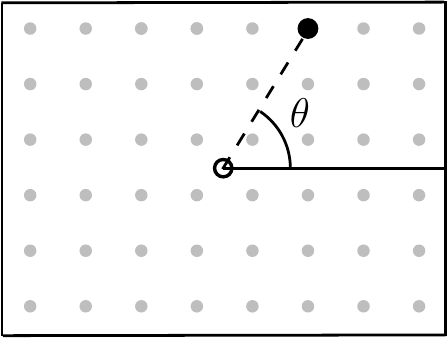}
	\caption{\small Subpartitioning scheme for a rectangular partition. Using the centroid of the partition as the release point, each node is assigned a quadrant angle using the release point and a reference axis as shown.}
	\label{fig:subpart}
\end{figure}

In order to generate a trajectory for a given subpartition, the path of minimum total Euclidean distance passing through all nodes in the subpartition $V_{i,k}$ is determined via the DFJ algorithm. The DFJ algorithm returns a sequence $\{\nu_{i,k,j}\}_{j=1:|V_{i,k}|}$ of nodes in $V_{i,k}$ sorted according to the minimum-distance trajectory.
% For a given partition size $(a_1,a_2)$ and a set of nodes $V_{i,k}$ in subpartition $k$ of partition $\mathcal{P}_i$, let $\overline{\overline{V}}_{g,i}$ be the result of the DFJ algorithm. 
Let $\tau_{j,k}$ be the time it takes a UAV to reach the $j^{th}$ point in the subpartition assigned to UAV $k$ from the center of the partition i.
For the line segment from $\nu_{i,k,j}$ to $\nu_{i,k,j+1}$, assuming the UAV reaches point $\nu_{i,k,j}$ at time $\tau_{j,k}$ and assuming the UAV is moving at maximum speed, the velocity of the UAV will be,
\begin{equation}\label{uav_traj}
    u_k^A(t) = 
    \begin{cases}
      u_{max}^A\frac{\nu_{i,k,j+1}-\nu_{i,k,j}}{\|\nu_{i,k,j+1}-\nu_{i,k,j}\|_2} & \text{if patrolling } \\
      u_k^G(t) & \text{otherwise}.
    \end{cases}
\end{equation}
where the UAV is in the patrolling mode during the times $t \in \bigcup_{\alpha=0}^{\infty}\left(t_i+\alpha T_c+\tau_{j,k} \ \scalebox{2}{,} \ t_i+\alpha T_c+\tau_{j+1,k} \right)$.
%\small
%\begin{equation}\label{uav_time}
%    \forall t \in \bigcup_{\alpha=0}^{\infty}\left(t_i+\alpha T_c+\tau_{j,k} \ \scalebox{2}{,} \ t_i+\alpha T_c+\tau_{j+1,k} \right) .
%\end{equation}
%\normalsize
%Otherwise, $u_k^A(t)$ is equal to the velocity of the UGV. As the trajectories start and end at the first node $\nu_{i,k,1}$, $\nu_{i,k,|V_{i,k}|+1}= \nu_{i,k,1}$.

Let $L_i$ be the the optimal cost of the DFJ algorithm when evaluating the optimal trajectory for UAV $i$. In order to ensure that the UAVs never run out of energy while patrolling, the UAVs must cover the total distance of its trajectory given their energy constraints and maximum speed. Assuming the UAVs move at maximum speed, the distance travelled can be expressed in terms of the energy, $\Delta e_i$ consumed by the UAV $i$, that is $L_i = u_{max}^A\frac{\Delta e_i}{\beta^-}$.
%\begin{equation}
%    L_i = u_{max}^A\frac{\Delta e_i}{\beta^-}
%\end{equation}
Let $\Delta e$ be the maximum amount of energy consumed by a UAV while covering the partition, that is,
\begin{equation}\label{covertime}
    \Delta e = \underset{i}{\text{max}} \left( \Delta e_i \right) = \underset{i}{\text{max}} \frac{L_i \beta^-}{u_{max}^A}.
\end{equation}
Accordingly, the UAVs can cover the partition in one fuel cycle if
\begin{equation}\label{econstraint}
    \underset{i}{\text{max}} \frac{L_i\beta^-}{u_{max}^A} \le \bar{e}.
\end{equation}
By enforcing this energy constraint, the UAVs will be capable of implementing the trajectories generated by the DFJ algorithm, and the supercycle will result in complete safe patrolling of the environment. 
\begin{algorithm}[htb!]\small
	\caption{Partition Feasibility Check}\label{alg_1}
	\textbf{INPUT: } $n,m, \bar{e},u_{max}^A,\beta^-,\bar{x},\bar{y}$ \\
	\textbf{OUTPUT: } $T_c,\Delta e,\{\mathcal{P}_i\}_{i=1:|\mathcal{P}|},\{\nu_{i,k,j}\}_{j=1:|V_{i,k}|} \ \forall k$
	\begin{algorithmic}[1]
	\For{$k=1:\frac{n}{m}$}
	\State \textrm{Subpartition $a_1\times a_2$ nodes via \eqref{subpartition}}
	\State \textrm{Use DFJ algorithm to find $\{\nu_{i,k,j}\}_{j=1:|V_{i,k}|}$}
	\State $\tau \gets 0$
	\For {$j=1:|V_{i,k}|$}
	\State $u_k^A(t) \gets$ \eqref{uav_traj}
	%\State $\tau \gets \tau + \frac{\|\nu_{i,k,j+1}-\nu_{i,k,j}\|_2}{u_{max}^A}$
	\EndFor
	\EndFor
	\If{ \textrm{\eqref{econstraint} is true}}
	\State Use DFJ algorithm to find $\{\mathcal{P}_i\}_{i=1:|\mathcal{P}|}$
	\State $\Delta e \gets$ \eqref{covertime}, $T_c \gets$ \textrm{\eqref{objective}}
	\Else
	\State $T_c \gets$ $\infty$ 
	\EndIf
\end{algorithmic}
\end{algorithm}
\linebreak

\begin{algorithm}\small
	\caption{UAV-UGV Trajectory Generation}\label{alg_5}
	\textbf{INPUT: } $n,m, p^A(0),p^G(0),u_{max}^A,u_{max}^G,\beta^-,\beta^+,\bar{x},\bar{y}$ \\
	\textbf{OUTPUT: } $\mathbf{u}^A,\mathbf{u}^G$
	\begin{algorithmic}[1]
	\State For all combinations of $a_1\le \bar{x}$ and $a_2 \le \bar{y}$,
    Run Alg. \ref{alg_1} 
	$T_c(a_1,a_2) \gets T_c$    
	\State $(a_1,a_2) \gets \underset{a_1,a_2}{argmin} (T_c) $
	\State Run Alg. \ref{alg_4}
	\State Run Alg. \ref{alg_1}
    \State $u^G(t)=0 \ \forall t$ 
    \State $t_i \gets 0$
    \For {$i=1:|\mathcal{P}|$}
    \For {$k=1:\frac{n}{m}$}
    \State $\tau \gets t_i$
    \For {$j=1:|V_{i,k}|$}
	\State $u_k^A(t) \gets$ \eqref{uav_traj}
	\State $\tau \gets \tau + \frac{\|\nu_{i,k,j+1}-\nu_{i,k,j}\|_2}{u_{max}^A}$
	\EndFor
	\EndFor
    \For {$t \in (t_i+\frac{\Delta e}{\beta^-}, t_i+\frac{\Delta e}{\beta^-}+\frac{\|c_{i+1}-c_i\|_2}{u_{max}^G})$}
    \State $u^G(t) \gets $ \eqref{ugv_traj}
    \EndFor
    \State $t_i \gets t_i + \Delta t_i$ from \eqref{eq:delta_time_partitions}
    \EndFor
\end{algorithmic}
\end{algorithm}

%\linebreak

\begin{lemma}\label{lemma1}
There will always exist at least one node that belongs to only one partition if Alg. \ref{alg_4} is used to generate the partitions of a rectangular environment.
\end{lemma}
\begin{proof}
    Consider the node $v_1$ in the lower left corner of the first partition $\mathcal{P}_1$ at the lower left corner of the environment. There are four possible cases for this partition.
    \begin{itemize}
        \item \textbf{$\bar{x}>a_1$ and $\bar{y}>a_2$: } In this case, $v_1$ does not lie in any other partition besides $\mathcal{P}_1$.
        \item \textbf{$\bar{x}=a_1$ and $\bar{y}>a_2$: } In Alg. \ref{alg_4}, if  $\bar{x}=a_1$, then the lines $3-7$ and $12-15$ will not be executed. In this case, $\mathcal{P}_1$ will be uniquely defined by the lines $8-11$, and therefore $v_1$ will only belong to one partition.
        \item \textbf{$\bar{x}>a_1$ and $\bar{y}=a_2$: } In this case, the lines $3-7$ and $8-11$ will be skipped. Now $\mathcal{P}_1$ is defined in the lines $12-15$, hence $v_1$ belongs to only one partition.
        \item \textbf{$\bar{x}=a_1$ and $\bar{y}=a_2$: } Now all of the for loops from lines $3-15$ will be skipped, and there is only one partition defined in lines $16-17$, which is $\mathcal{P}_1$.
    \end{itemize}
    \vspace{-5mm}
\end{proof}

\begin{proposition}
If a single UAV-UGV team follows the same supercycle over a set of equal rectangular partitions according to Alg. \ref{alg_1}, then the maximum age after a period $T_c$ will be the period of the supercycle, i.e.,
\begin{equation}
    \bar{T}(t) = T_c, \quad \forall t \ge T_c.
\end{equation}
\end{proposition}
\begin{proof}
In light of Lemma \ref{lemma1}, there will always be at least one node $v\in V$ that will only be contained in a single partition. The DFJ algorithm in Alg. \ref{alg_1} guarantees that trajectories in $\mathcal{P}_i$ over a single fuel cycle will visit every node in $\mathcal{P}_i$ only once when $\mathcal{P}_i$ is visited by the UAV-UGV team. Finally, a UAV-UGV team following a supercycle over the partitions will visit each partition only once. These three results imply that there exists at least one node that is visited only once per supercycle. Since the period of the supercycle is $T_c$, the time between visits to such a node will also be $T_c$. Therefore, the maximum age will be $T_c$ for all times $t\ge T_c$.
\end{proof}

Considering that $T_c$ will be equivalent to the maximum age in the limit $t\to \infty$, we formulate an optimization problem that will be equivalent to \eqref{eq:problem}. The choice of $(a_1,a_2)$ defines $\mathcal{P}$, and therefore $T_c$. The goal is to minimize $T_c$ over all pairs $(a_1,a_2)$ subject to the energy limitations of the UAVs,
\begin{align}\label{eq:problem2}
    \min_{a_1,a_2} & \quad {T_c}(a_1,a_2), \\
    s.t. & \quad \eqref{econstraint}.
\end{align}

\begin{proposition}\label{prop1}
For any $\frac{n}{m}$ and $(a_1,a_2)$, if Alg. \ref{alg_1} returns UAV trajectories, they start and end at the centroid of the partition, visit all nodes in the partition once, and have lengths that can be traversed with an energy less than or equal to $\bar{e}$. 

% For a subpartitioning scheme of an odd-dimensioned square partition wherein the boundaries are defined by lines connecting the center of the partition to the outermost layer of the partition, a trajectory $\pi$ generated by Alg. \ref{alg_2} always exists for which $f(a)$ in Equation \ref{trajbound} is a valid upper bound on $|\pi_i|$.
\end{proposition}

\begin{proof}
    The DFJ algorithm in line $3$ of Alg. \ref{alg_1} guarantees that the trajectories collectively pass through all nodes in the partition, and begin and end at the release point. Line $10$ of Alg. \ref{alg_1} guarantees that the UAVs will return to the release point before running out of energy.
\end{proof}

\subsubsection{Deployment Protocol}
We consider homogeneous UAV-UGV teams that are deployed at uniformly spaced intervals along a given supercycle.
In line 1 of Alg. \ref{alg_3}, all UAV-UGV teams are deployed initially to $c_1$. However, the UAVs are not released from UGV $i$ until time $(i-1)T_c/m$, at which point the supercycle is initiated. Therefore, teams $i=1,2,...m$ will begin the supercycle at times $t=\{0,T_c/m,...,(m-1)T_c/m\}$. 

\begin{algorithm}[H]\small
	\caption{Deployment Protocol}\label{alg_3}
	\textbf{INPUT: } $n,m, p(0),q(0),u_{max}^A,u_{max}^G,\beta^-,\beta^+,\bar{x},\bar{y}$
	\begin{algorithmic}[1]
	\State Run Alg. \ref{alg_5}
	\State \textrm{Deploy all UGVs to  $c_1$}.
	\State \textrm{$t \gets 0$}.
	\State \textrm{$u_1^G(t) \gets$ \eqref{ugv_traj}}
	\State \textrm{$u_{1,k}^A(t) \gets$ \eqref{uav_traj} $\ \forall  k=1,...,n/m$}
	\State \textrm{$u^G_j(t) \gets u^G_1(t-(j-1)\frac{T_c}{m}) \quad \forall j=2,...,m$}
	\State \textrm{$u^A_{j,k}(t) \gets u^A_{1,k} (t-(j-1)\frac{T_c}{m} ) \quad \forall j=2,...,m, \ \forall k=1,...,\frac{n}{m}$}
\end{algorithmic}
\end{algorithm}

\begin{proposition}
Given a supercycle of period $T_c$ over a set of partitions $\mathcal{P}$, if $m$ UAV-UGV teams follow the deployment protocol of Alg. \ref{alg_3}, the maximum age over all nodes $v\in V$ starting after a time $T_c$ has passed will be $\frac{T_c}{m}$
\begin{equation}
    \bar{T}(t) = \frac{T_c}{m}, \quad \forall t \ge T_c.
\end{equation}
\end{proposition}
\begin{proof}
The teams are identical, follow the same deterministic dynamics and start with the same initial condition. Therefore, when the teams are released from the same location with $T_c/m$ time difference as in Alg. \ref{alg_3}, that time difference will always hold. More specifically, for a given team $k$ with UAV $j$, where $j \in \{1,2, \hdots, n/m\}$  denotes the subpartition assigned to the UAV,
\begin{equation}
    p^G_k (t)= p^G_{k+1}(t+T_c/m), 
\end{equation}
\begin{equation}
p^A_{k,j}(t)= p^A_{k+1,j}(t+T_c/m),
\end{equation}
\begin{equation}
e_{k,j}(t)= e_{k+1,j}(t+T_c/m), 
\end{equation}
for all $t\ge 0$. 
Consider a node $v\in V$ that is only visited once per supercycle. The existence of such a node is proven in Lemma \ref{lemma1}. Suppose that $v$ is visited at time $\tau$ during the supercycle that begins at time $t=0$. Given that all teams $i=1,...,m$ begin at $c_1$, and initiate the same supercycle, then the node $v$ will be visited by teams $i=1,2,...,m$ at times $t=\{\tau,\tau+\frac{T_c}{m},...,\tau+(m-1)\frac{T_c}{m}\}$. The spacing of all of these times is $\frac{T_c}{m}$, which will be the 
maximum age for any node $v\in V$ visited only once per supercycle.
Given the partitioning Alg. \ref{alg_4} and the trajectory generation Alg. \ref{alg_5}, all nodes will be visited at least once per supercycle. For all nodes visited more than once per supercycle, the maximum age will be lower. Therefore, the maximum age over all nodes after a time $T_c$ has passed will be $\frac{T_c}{m}$. \end{proof}

\

\vspace{-7mm}
\section{Simulation Results}
Consider $3$ UGVs and $15$ UAVs that can form homogeneous UAV-UGV team of $1$ UGV and $5$ UAVs. Suppose that each UAV has a maximum energy of $\bar{e}=100$ and a detection footprint of $d=33$. The environment is assumed to have dimensions $x_{max}=1584$ and $y_{max}=1056$, which imply $\bar{x}=48$ and $\bar{y}=32$, as depicted in Fig. \ref{fig:simulation}. We assume the UGV transport rate $u_{max}^G=5$, and the charging and depletion rates for the UAVs are $\beta^+=\beta^-=0.5$.

\begin{figure}[h]
    \centering
    \includegraphics[scale=0.23]{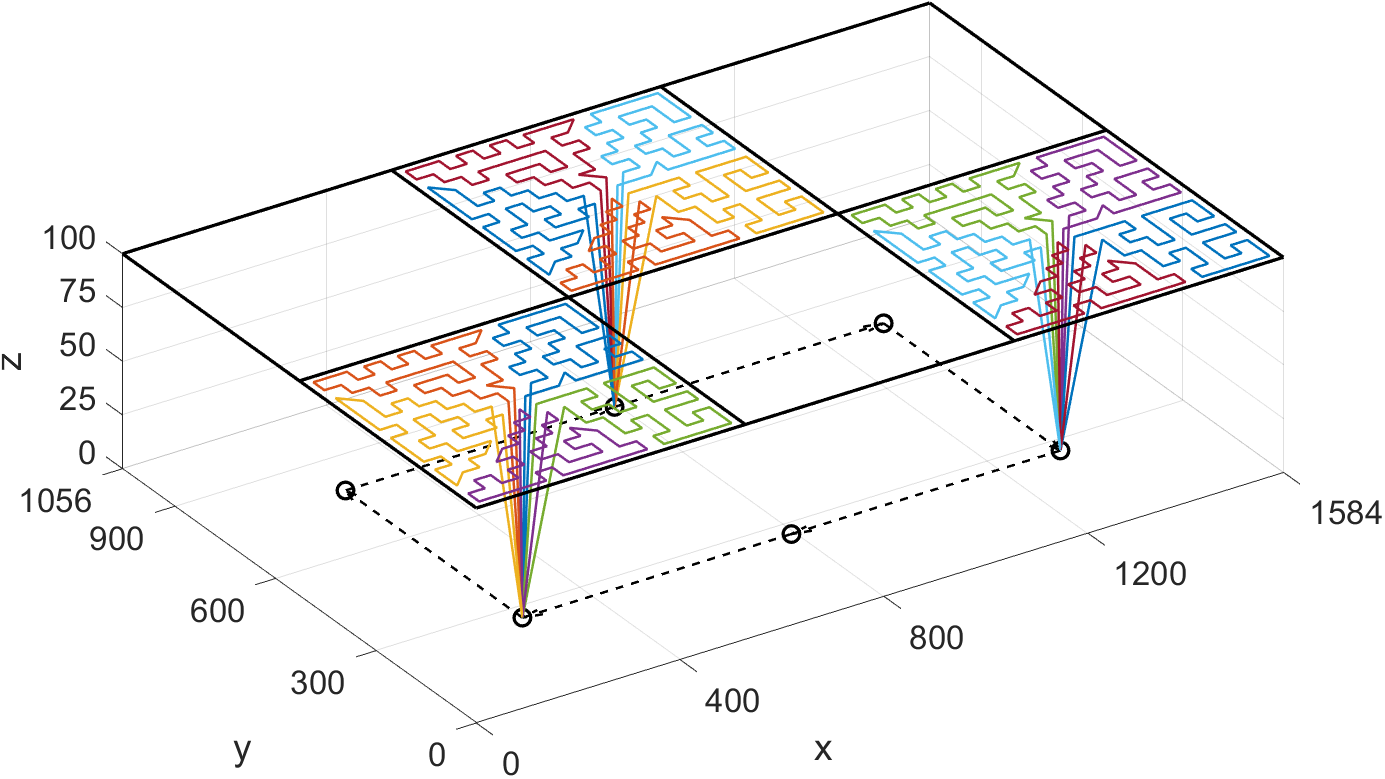}
    \caption{\small Trajectories for three homogeneous teams over a single fuel cycle. The UAV-UGV teams are positioned across a supercycle with equal temporal spacing.}
    \label{fig:simulation}
\end{figure}

Based on Alg. \ref{alg_5}, the optimal partition dimensions are $(a_1,a_2)=(16,16)$, implying that $|\mathcal{P}|=6$.%, per \eqref{partsize}. 
The corresponding value of $\Delta e$ is $96.07$ per \eqref{covertime}. Alg. \ref{alg_5} results in a supercycle period of $T_c=2306$. The resulting trajectories are depicted in Fig. \ref{fig:simulation}, wherein the three UAV-UGV teams are positioned on the supercycle with equal temporal spacing, resulting in a minimum long-term maximum age of $T_c/m=768.6$. 
% As expected, the inequalities of Theorem \ref{theorem2} are satisfied. 
% Considering that we set the UGV mobility at $u=1$ and the charging rate at $\beta=1$, we can see from Equation \ref{multicycle} that the UGV transport time played no role in the performance. 

%\begin{algorithm}
%	\SetAlgoLined
%	\caption{Optimal Path Solver (UGV $j$)}
%	\textbf{Input Parameters: } $\bar{n}, \mathcal{G},V_p, \bar{e}$

%\end{algorithm}

\section{Conclusion}
In this paper, a strategy for implementing persistent surveillance over a rectangular environment was presented for multiple energy-constrained UAVs supported by multiple mobile charging stations. This strategy reduces the problem of optimizing $n+m$ trajectories under energy constraints to the simpler problem of optimizing a supercycle of partitions for a single team of $\frac{n}{m}$ UAVs and one UGV. We showed that the UAVs will always be safe (not running out of energy) and the proposed strategy will result in a minimum long-term maximum age based on the length of the supercycle and the number of teams. As a future work, we plan to extend this study by considering heterogeneous teams and increasing the complexity of the environment. %In future works, a fundamental performance limit could be developed to gauge the performance of this method, and provide a benchmark for strategies for the energy constrained multi-UAV persistent surveillance problem.

\bibliography{ref}

\end{document}